\documentclass[11pt]{article}

\newcommand{\version}{long}

\usepackage{algorithm,algpseudocode,amsmath,amssymb,amsthm,complexity,etoolbox,fullpage,graphicx,ifthen,mathtools,url,thmtools,thm-restate}
\usepackage[inline, shortlabels]{enumitem}

\usepackage{complexity}

\newclass{\ioPSPACE}{i.o.\text{-}PSPACE}
\newlang{\Halt}{Halt}
\usepackage{etoolbox,ifthen,tabulary}

\newcommand{\ifft}{\ifundef{\shortiff}{if and only if }{iff }}

\newcommand{\indraft}[1]{\ifthenelse{\equal{\version}{draft}}{#1}{}}
\newcommand{\infinal}[1]{\ifthenelse{\equal{\version}{final}}{#1}{Only shown in the final version}}
\newcommand{\inshort}[1]{\ifthenelse{\equal{\version}{short}}{#1}{}}
\newcommand{\inlong}[1]{\ifthenelse{\equal{\version}{long}}{#1}{}}

\newcommand{\secref}[1]{Section~\ref{sec:#1}}
\newcommand{\ssecref}[1]{Subsection~\ref{ssec:#1}}
\newcommand{\tabref}[1]{Table~\ref{tab:#1}}

\newcommand{\renameenv}[2]{
  \expandafter\let\csname #1#2\expandafter\endcsname
  \csname #1\endcsname
  \expandafter\let\csname end#1#2\expandafter\endcsname
  \csname end#1\endcsname
  \expandafter\let\csname #2\endcsname\relax
  \expandafter\let\csname end#2\endcsname\relax}

\ifundef{\defaultlists}{
  \usepackage[inline,shortlabels]{enumitem}
  \setenumerate[1]{(a),itemsep=0pt,topsep=3pt,parsep=0pt,partopsep=0pt}
  \setenumerate[2]{(i),noitemsep,topsep=3pt,parsep=0pt,partopsep=0pt}
  \setenumerate[3]{(A),noitemsep,topsep=3pt,parsep=0pt,partopsep=0pt}
  \setenumerate[4]{(I),noitemsep,topsep=3pt,parsep=0pt,partopsep=0pt}
  \setitemize{noitemsep,topsep=3pt,parsep=0pt,partopsep=0pt}
  \setdescription{noitemsep,topsep=3pt,parsep=0pt,partopsep=0pt}
  \setlist{noitemsep,topsep=3pt,parsep=0pt,partopsep=0pt}}{}

\newcolumntype{x}[1]{>{\centering\arraybackslash}m{#1}}
\usepackage{algpseudocode,amsmath,amssymb,etoolbox,fixltx2e,mathtools}

\let\eqref\relax

\DeclareFontFamily{U}{mathx}{\hyphenchar\font45}
\DeclareFontShape{U}{mathx}{m}{n}{
      <5> <6> <7> <8> <9> <10>
      <10.95> <12> <14.4> <17.28> <20.74> <24.88>
      mathx10
      }{}
\DeclareSymbolFont{mathx}{U}{mathx}{m}{n}
\DeclareMathSymbol{\bigtimes}{1}{mathx}{"91}

\algnewcommand{\Input}{\item[\textbf{Input:}]}
\algnewcommand{\Output}{\item[\textbf{Output:}]}
\MakeRobust{\Call} 

\newcommand\restr[2]{{
  \left.\kern-\nulldelimiterspace
  #1
  \vphantom{\big|}
  \right|_{#2}
  }}
\newcommand{\concat}{:}
\newcommand{\indphi}[1]{\restr{\phi}{#1}}
\newcommand{\abs}[1]{\left\vert#1\right\vert}

\newcommand{\coleq}{\coloneqq}

\newcommand{\soc}{\mathrm{soc}}

\newcommand{\genb}[2]{\left\langle#1 \;\middle|\; #2\right\rangle}

\newcommand{\nth}[1]{\ensuremath{{#1}^{\mathrm{th}}}}

\renewcommand{\algref}[1]{Algorithm~\ref{alg:#1}}

\newcommand{\linref}[1]{line~\ref{line:#1}}

\newcommand{\thmref}[1]{Theorem~\ref{thm:#1}}
\newcommand{\lemref}[1]{Lemma~\ref{lem:#1}}

\newcommand{\eqref}[1]{(\ref{eq:#1})}

\newcommand{\bmg}{\mathbf{g}}
\newcommand{\bmh}{\mathbf{h}}

\newcommand{\bmr}{\mathbf{r}}
\newcommand{\bms}{\mathbf{s}}

\newcommand{\bmu}{\mathbf{u}}
\newcommand{\bmv}{\mathbf{v}}

\newcommand{\ra}{\rightarrow}

\newcommand{\tril}{\triangleleft}

\inshort{\usepackage{amsthm}
\makeatletter
\def\thm@space@setup{\thm@preskip=3pt \thm@postskip=3pt}
\makeatother

\makeatletter
\renewenvironment{proof}[1][\proofname]{\par
  \pushQED{\qed}
  \normalfont
  \topsep3pt \partopsep0pt 
  \trivlist
  \item[\hskip\labelsep
        \itshape
    #1\@addpunct{.}]\ignorespaces
  }{
    \popQED\endtrivlist\@endpefalse
    \addvspace{0pt plus 0pt} 
  }
\makeatother
} 

\usepackage{etoolbox,ifthen,thmtools,thm-restate}

\ifundef{\dontnumberwithin}{\declaretheorem[numberwithin=section]{dummy}}{\declaretheorem{dummy}} 
\declaretheorem[sibling=dummy]{theorem}

\declaretheorem[sibling=dummy]{lemma}

\ifundef{\defaultthmcontinues}{\renewcommand{\thmcontinues}[1]{}}{}

\inshort{\newcommand{\shortiff}{}}
\inshort{\usepackage[compact]{titlesec}}

\newcommand{\bigenenum}{\text{BIDIRECTIONAL-GENERATOR-ENUMERATION}}
\newcommand{\insexts}{\text{INSERT-EXTENSIONS}}
\newcommand{\genenum}{\text{GENERATOR-ENUMERATION}}
\newcommand{\compsA}{\text{COMPOSITION-SERIES-ALICE}}
\newcommand{\compsB}{\text{COMPOSITION-SERIES-BOB}}
\newcommand{\comps}{\text{COMPOSITION-SERIES}}
\newcommand{\mns}{\text{MINIMAL-NORMAL-SUBGROUPS}}
\newcommand{\sims}{\text{SIMPLE-SUBGROUPS}}

\begin{document}

\title{Bidirectional Collision Detection and Faster Deterministic Isomorphism Testing}
\author{David J. Rosenbaum \\ University of Washington, \\ Department of Computer Science \& Engineering \\ Email: djr@cs.washington.edu}
\date{May 16, 2013}

\maketitle
\thispagestyle{empty}

\begin{abstract}
  In this work, we introduce bidirectional collision detection --- a new algorithmic tool that applies to the collision problems that arise in many isomorphism problems.  For the group isomorphism problem, we show that bidirectional collision detection yields a deterministic $n^{(1 / 2) \log n + O(1)}$ time algorithm whereas previously the $n^{\log n + O(1)}$ generator-enumeration algorithm was the best result for several decades.  For the hard special case of solvable groups, we combine bidirectional collision detection with methods from the author's previous work to obtain a deterministic square-root speedup over the best previous algorithm.  We also show a deterministic square-root speedup over the best previous algorithm for testing isomorphism of rings.  We can even apply bidirectional collision detection to the graph isomorphism problem to obtain a deterministic $T^{1 / \sqrt{2}}$ speedup over the best previous deterministic algorithm.  Although the space requirements for our algorithms are greater than those for previous deterministic isomorphism tests, we show time-space tradeoffs that interpolate between the resource requirements of our algorithms and previous work.
\end{abstract}

\inlong{
  \newpage
  \setcounter{page}{1}}

\section{Introduction}

In an isomorphism problem, we are given two algebraic or combinatorial objects and must decide if they have the same structure.  In this work, we study several different isomorphism problems and show speedups over the best previous deterministic algorithms.

We start with the \emph{group isomorphism problem} where we are given multiplication tables representing two groups $G$ and $H$ of order $n$ and must decide if $G \cong H$.  While polynomial-time algorithms are known for a number of special cases of group isomorphism~\cite{lipton1977a,savage1980a,vikas1996a,kavitha2007a,legall2008a,qiao2011a,babai2011a,codenotti2011a,babai2012a,babai2012b}, the $n^{\log_p n + O(1)}$ generator-enumeration algorithm~\cite{felsch1970a,lipton1977a} (cf.~\cite{miller1978a}) has been the most efficient method for testing isomorphism of general groups since its discovery in the 1970's.  The algorithm is based on the fact that every group of $n$ elements has an ordered generating set $\bmg$ of size at most $\log_p n$ (where $p$ is the smallest prime dividing the order of the group).  Since any isomorphism from $G$ to $H$ can be specified by the images of the elements of any ordered generating set of $G$, one can determine if $G$ and $H$ are isomorphic by considering each of the $n^{\log_p n}$ maps from $\bmg$ into $H$ and testing in polynomial time if it specifies an isomorphism.  Improving the generator-enumeration algorithm for general groups was recently noted as a longstanding open problem~\cite{lipton2011b}.


In previous work~\cite{rosenbaum2012b}, the author showed a relationship between group isomorphism and collision detection and utilized it to obtain a randomized square-root speedup over the generator-enumeration algorithm\footnote{This relationship is in fact more general and also applies to many other isomorphism problems.}.  In this work, we introduce the \emph{bidirectional collision detection} framework --- a new technique for deterministically obtaining square-root speedups for isomorphism problems.  Since bidirectional collision detection in particular applies to the class of collision problems that arise in group isomorphism, we obtain a deterministic square-root speedup over the best previous algorithm for general groups.




\begin{restatable}{theorem}{gengroupiso}
  \label{thm:gen-group-iso}
  General group isomorphism is in $n^{(1 / 2) \log_p n + O(1)}$ deterministic time where $p$ is the smallest prime dividing the order of the group.
\end{restatable}

The main result of~\cite{rosenbaum2012b} was a deterministic square-root speedup over the generator-enumeration algorithm for the hard special case of solvable group isomorphism; the present paper extends the improved bound of~\cite{rosenbaum2012b} to general groups.  Since the techniques used in~\cite{rosenbaum2012b} are independent from our bidirectional collision detection method, we can combine them with bidirectional collision detection to obtain a deterministic fourth-root speedup over the generator-enumeration algorithm for the class of solvable groups.


\begin{restatable}{theorem}{solgroupiso}
  \label{thm:sol-group-iso}
  Solvable-group isomorphism is in $n^{(1 / 4) \log_p n + O(\log n / \log \log n)}$ deterministic time where $p$ is the smallest prime dividing the order of the group.
\end{restatable}

Bidirectional collision detection can also be applied to the ring isomorphism problem to obtain a square-root speedup over the natural analogue of the generator-enumeration algorithm for rings.

In the case of graph isomorphism, bidirectional collision detection can used to obtain a deterministic $T^{1 / \sqrt{2}}$ speedup over the previous best algorithm~\cite{babai1983b}.


While most algorithms for isomorphism problems can be implemented in polynomial space, our algorithms require space roughly equal to their runtimes.  By breaking the underlying bidirectional collision problem up into blocks, we show that the generator-enumeration algorithm and our algorithm are extreme points of the time-space tradeoff $TS = n^{\log_p n + O(1)}$ for general group isomorphism.  We also generalize \thmref{sol-group-iso} to obtain a time-space tradeoff of $TS = n^{(1 / 2) \log_p n + O(\log n / \log \log n)}$ for solvable groups.  While randomized analogues of our algorithms exist, there is currently no benefit to using randomness as the time and space requirements are the same.  On the other hand, using a modification of the quantum algorithm for collision detection~\cite{bassard1997b}, we obtain quantum time-space tradeoffs of $T \sqrt{S} = n^{(1 / 2) \log_p n + O(1)}$ for general groups and $T \sqrt{S} = n^{(1 / 4) \log_p n + O(\log n / \log \log n)}$ for solvable groups.  Analogous time-space tradeoffs exist for rings and graphs. 

We start by describing our bidirectional collision-detection speedup in \secref{collision}.  In \secref{gen-algorithm}, we combine bidirectional collision detection with the generator-enumeration algorithm to obtain our deterministic square-root speedup for general group isomorphism.  In \secref{ring-iso}, we discuss a deterministic speedup for the ring isomorphism problem.  We show a deterministic fourth-root speedup for solvable-group isomorphism in \secref{p-sol-algorithm}.  In \secref{graph-iso}, we show our speedups for the graph isomorphism problem.  We generalize our algorithms by giving time-space tradeoffs in \secref{time-space}.  We conclude with the current state of the art and open problems in \secref{conclusion}.

\section{Bidirectional collision detection}
\label{sec:collision}
In this section, we describe a class of collision problems for which a square-root speedup can be obtained deterministically over the naive brute force algorithm.  Later in our paper, we shall show how this speedup can be used to obtain faster deterministic algorithms for several isomorphism problems. 

\subsection{Collision detection as a game}
Consider a game in which Alice and Bob have access to a tree $T$ with root $r$ via separate oracles.  For each node $x$, Alice and Bob have labels $\lambda_a(x)$ and $\lambda_b(x)$.  Assuming Alice knows her label $\lambda_a(x)$ for a node $x$, she can compute her labels for the children of $x$ using her oracle; similarly, if Bob knows his label $\lambda_b(x)$ for a node $x$, he can compute his labels for the children of a node $x$ using his oracle.  In our game, Alice and Bob each start with their label for the root node.  Alice then selects a set $A$ of simple paths from the root to the leaves specified in terms of her labels for the nodes.  Similarly, Bob selects a set $B$ of simple paths from the root to the leaves specified using his labels for the nodes.  No communication is allowed between Alice and Bob at any time.  Alice and Bob win if the sets $A$ and $B$ contain a common path (i.e., there exist $(v_1, \ldots, v_m) \in A$ and $(u_1, \ldots, u_m) \in B$ such that each $\lambda_a^{-1}(v_i) = \lambda_b^{-1}(v_i)$).  It is easy for Alice and Bob to win by choosing all possible paths from the root to the leaf nodes; however, they wish to win while also keeping the total number of paths $\abs{A} + \abs{B}$ as small as possible.

We consider a winning strategy for Alice and Bob.  First, we define several parameters that allow us to characterize the trees for which our strategy is efficient.  Let $N$ be the number of nodes in the tree $T$ and let $T_1$ be the subtree of $T$ consisting of nodes at a distance of at most $d$ from the root of $T$.  Suppose that $T$ has the properties that

\begin{enumerate}
\item there is at least one node in $T$ at a distance of $d$ from the root,
\item there are at most $f(N)$ nodes in the subtree $T_1$ and
\item every subtree $T_2$ rooted at a node of distance $d$ from the root contains at most $g(N)$ nodes
\end{enumerate}

For such trees, we show a winning strategy for Alice and Bob such that $\abs{A} + \abs{B} \leq f(N) + g(N)$.  We call our method bidirectional collision detection due to an analogy with bidirectional search\footnote{Bidirectional search is a method for finding a path which searches from both the source and destination vertexes and meets in the middle.}.  Alice starts by computing the set of all simple paths to leaf nodes of the subtree $T_1$; she then extends each of these paths to a simple path to a leaf node of $T$ arbitrarily.  Placing the resulting paths in $A$, we see that $\abs{A} \leq f(N)$.  Bob starts by selecting a simple path $\bmv = (v_1, \ldots, v_d)$ from the root to any node $v_d$ at distance $d$.  Clearly, $v_d$ is in the subtree $T_1$.  Bob proceeds by collecting all possible simple paths of minimal length from $v_d$ to leaf nodes of $T$ into the set $B'$.  By concatenating $\bmv$ with each path in $B'$, he obtains a set $B$ of paths from the root of $T$ to its leaves; $B$ has the property that $\abs{B} \leq g(N)$.

We now argue that Alice and Bob have a path in common.  It is clear that the path $\bmv$ used by Bob as a prefix of each of his paths agrees with one of Alice's paths $\bmu = (u_1, \ldots, u_m)$ on the first $d$ nodes since Alice considers all simple paths from the root to the leaf nodes of $T_1$.  Let $T_2$ be the subtree of $T$ rooted at $v_d$.  Since Bob considers all possible minimal length simple paths from $v_d$ to the leaves of $T$, it is clear that one of his paths will agree with Alice's path $\bmu$.  We have already seen that $\abs{A} + \abs{B} \leq f(N) + g(N)$.  Later, we shall see in~\secref{gen-algorithm} how bidirectional collision detection can be applied to obtain a square-root speedup for general group isomorphism.

We note that as described above, it is implicitly assumed that there is always a common path in Alice and Bob's labelings of the tree $T$.  In the problems we consider in this paper, Alice and Bob have labelings for the trees $T_a$ and $T_b$ respectively and are promised that either $T_a$ and $T_b$ have disjoint sets of leaf nodes or $T_a = T_b$.  As we shall see, this corresponds precisely to determining if there is an isomorphism.  Our task is therefore to decide if $T_a = T_b$.  We accomplish this using the strategy described above.  If $T_a = T_b$, then the sets $A$ and $B$ contain a common path.  Otherwise, the sets $A$ and $B$ are disjoint.  We can thus decide if $T_a = T_b$ by checking which of these is the case.

\subsection{A deterministic time-space tradeoff for collision-detection}
\label{ssec:d-time-space}
We now discuss time-space tradeoffs for collision problems.  First of all, it suffices to store only the leaf node at which each of the paths in $A$ and $B$ terminates and then test if the resulting sets of leaf nodes are disjoint; this yields a small reduction in the amount of space required.  The arguments of the preceding subsection then imply an $O((f(N) + g(N)) s)$ time algorithm (where $s$ is the amount of space required to store Alice or Bob's label for a node in the tree) for any collision problem that fits into the tree framework described.  (We shall henceforth refer to this as the \emph{bidirectional collision detection} algorithm.)  Since $f(N)$ and $g(N)$ will typically be $O(\sqrt{N} \poly(s))$, this constitutes a square-root speedup over the naive brute-force algorithm for many problems.  However, since bidirectional collision detection must store all of the leaf nodes that correspond to the paths in the sets $A$ and $B$, it requires $O((f(N) + g(N)) s)$ space whereas the brute-force algorithm typically needs only $\poly(s)$ space (which will normally be much smaller).

Both the naive brute-force algorithm and the bidirectional collision detection algorithm are special cases of a general time-space tradeoff.  Here, one arbitrarily breaks the sets $A$ and $B$ up into chunks of $\Delta$ elements; we then test if $A$ and $B$ contain a common element by checking if any of the at most $\frac{f(N) g(N)}{\Delta^2}$ pairs of chunks contain a common element.  This takes time $T = \frac{f(N) g(N)}{\Delta} \poly(s)$.  Here, we have assumed that enumerating the sets $A$ and $B$ can be done at a cost of $\poly(s)$ per element; this will be the case in all of the problems considered in this paper.  Since the amount of space required is $S = \Delta s$, we obtain a deterministic time-space tradeoff of $TS = f(N) g(N) \poly(s)$ where $S \leq \min\{f(N), g(N)\}$.  Both bidirectional collision detection and the brute-force algorithm are special cases of this time space tradeoff if we neglect small differences in the runtime (which correspond to lower order terms in the problems we consider).

We remark that there is currently no advantage for randomized algorithms over deterministic algorithms due to our bidirectional collision detection techniques.  The next model of computation to consider is therefore quantum algorithms. 

\subsection{A quantum time-space tradeoff for collision detection}
\label{ssec:q-time-space}
One can also obtain a quantum time space tradeoff for collision problems that fit into the tree framework.  Quantum algorithms for collision detection typically assume that there are many collisions whereas in our bidirectional collision detection strategy we only have a single collision.  We remedy this by considering the full tree for both Alice and Bob.  Fortunately, the $N^{1 / 3} \poly(s, \log N)$ runtime of this quantum algorithm~\cite{bassard1997b} is in terms of the size $N$ of a single tree rather than all pairs of nodes.  Thus, the quantum algorithm inherently handles the two separate copies of the tree without any degradation in performance.  While this algorithm uses $O(N^{1 / 3} \poly(s, \log N))$ space, modifying the algorithm to use only $S$ space in the obvious way results in a runtime of $T = \sqrt{N / S} \poly(s, \log N)$ where $S \leq N^{1 / 3}$.  Thus, we have a time-space tradeoff of $T \sqrt{S} = \sqrt{N} \poly(s, \log N)$ for quantum algorithms.


\section{Bidirectional generator enumeration}
\label{sec:gen-algorithm}
In the author's previous paper~\cite{rosenbaum2012b}, a relationship between group isomorphism and the collision problem was given.  Once adapted to the present setting, this result may be stated as follows.

\begin{lemma}[\cite{rosenbaum2012b}]
  \label{lem:gen-coll}
  Let $G$ and $H$ be groups and suppose that we have a deterministic algorithm with the following properties:

  \begin{enumerate}
  \item The algorithm either decides if $G \cong H$ or outputs two sets $A$ and $B$.  The set $A$ contains ordered generating sets for $G$ and the set $B$ contains ordered generating sets for $H$.
  \item The algorithm runs in time $t(n)$.
  \item If $\phi : G \ra H$ is an isomorphism, then there exist $\bmg \in A$ and $\bmh \in B$ such that $\phi(\bmg) = \bmh$.
  \end{enumerate}

  Then we can construct a deterministic algorithm that decides if two groups $G$ and $H$ are isomorphic in $t(n) \poly(n)$ time.
\end{lemma}

Our goal shall therefore be to exhibit an algorithm that satisfies the conditions of~\lemref{gen-coll} and runs in $n^{(1 / 2) \log_p n + O(1)}$ time where $p$ is the smallest prime which divides the order of the group.  We accomplish this by applying bidirectional collision detection to the problem of constructing an ordered generating set.

Consider two groups $G$ and $H$ and let $p$ be the smallest prime which divides the order of the groups.  We note that the crucial condition in \lemref{gen-coll} is (c); moreover, to satisfy (c), it suffices to consider the case where $G \cong H$.  We may view this as the case where there is an abstract group $K$ that is isomorphic to the groups $G$ and $H$ (which are specified by their multiplication tables).  We now relate this to the bidirectional collision framework given in \secref{collision} by defining the variables used in that section with respect to the present context.  We start by setting the distance parameter $d = (1 / 2) \log_p n$.  The nodes of the tree $T$ correspond to ordered generating sets of subgroups of $K$.  The root node of $T$ is the empty ordered generating set $()$; the children of the node $(x_1, \ldots, x_j) \in K^j$ are the vectors $(x_1, \ldots, x_{j + 1})$ with $x_{j + 1} \in K \setminus \genb{x_i}{1 \leq i \leq j}$.  The leaf nodes are the nodes $(x_1, \ldots, x_j) \in K^j$ such that $\genb{x_i}{1 \leq i \leq j} = K$.  In this way, descending from the root node of $T$ corresponds to constructing an ordered generating set of $K$.

\begin{algorithm}[H]
  \begin{algorithmic}[1]
    \Input Two groups $G$ and $H$ specified by their multiplication tables
    \Output Either decide if $G \cong H$ or return sets $A$ and $B$ that satisfy the conditions of \lemref{gen-coll}
    \Function{$\bigenenum$}{$G, H$}
      \State Let $p$ be the smallest prime that divides $n$ \label{line:set-p}
      \State $d \coleq (1 / 2) \log_p n$
      \If{$G$ has a generating set of size at most $d$} \label{line:gen-set-if}
        \State \Return $\Call{\genenum}{G, H}$
      \ElsIf{$H$ has a generating set of size at most $d$}
        \State \Return ``$G \not\cong H$''
      \EndIf \label{line:end-gen-set-if}
      \State $A' \coleq \{\}$ \label{line:set-A'}
      \State \Call{\insexts}{$(), G, A', 0, d$} \label{line:exts-G}
      \State $A \coleq \{\}$ \label{line:set-A}
      \For{$\bmg_1 \in A'$} \label{line:g1-loop}
        \State $\bmg_2 \coleq ()$
        \Repeat \label{line:g2-loop}
          \State Choose an arbitrary $g \in G \setminus \langle \bmg \rangle$
          \State $\bmg_2 \coleq \bmg_2 \concat (g)$
          \State $\bmg \coleq \bmg_1 \concat \bmg_2$
        \Until{$G = \langle \bmg \rangle$}
        \State Insert $\bmg$ into $A$ 
      \EndFor
      \State $\bmh_1 = ()$
      \For{$i = 1, \ldots, d$} \label{line:h1-loop}
        \State Choose an arbitrary $h \in H \setminus \langle \bmh_1 \rangle$
        \State $\bmh_1 \coleq \bmh_1 \concat (h)$
      \EndFor
      \State{$B \coleq \{\}$}
      \State \Call{\insexts}{$\bmh_1, H, B, 0, \infty$} \label{line:exts-H}
      \State Output $A$ and $B$
    \EndFunction
  \end{algorithmic}
  \caption{An algorithm for general group isomorphism}
  \label{alg:gen-AB}
\end{algorithm}

As in \secref{collision}, the nodes of the tree $T$ are not directly accessible.  However, we have access to $T$ through Alice and Bob's labelings $\lambda_a : K^* \ra G^*$ and $\lambda_b : K^* \ra H^*$.  Here, there are isomorphisms $\phi_a : K \ra G$ and $\phi_b : K \ra H$ such that $\lambda_\ell(x_1, \ldots, x_j) = (\phi_\ell(x_1), \ldots, \phi_\ell(x_j))$ for $\ell \in \{a, b\}$.  Note that there are less than $f(n) = n^{(1 / 2) \log_p n + O(1)}$ nodes at a distance of at most $d$ from the root of $T$.  Similarly, each subtree $T_2$ rooted at a node at distance $d$ from the root contains at most $g(n) = n^{(1 / 2) \log_p n + O(1)}$ nodes.  The argument of the previous section then implies that bidirectional collision detection yields sets $A$ and $B$ that satisfy the requirements of \lemref{gen-coll}.  In this way we obtain an $n^{(1 / 2) \log_p n + O(1)}$ algorithm for testing isomorphism of general groups.  The pseudocode for computing $A$ and $B$ is shown in \algref{gen-AB}; the symbol $\concat$ denotes vector concatenation.

\begin{algorithm}[H]
  \begin{algorithmic}[1]
    \Input A vector $\bmg_1 = (g_1, \ldots, g_k)$ of elements of a group $g$, a set $A$, the number $j$ of elements added to $\bmg_1$ so far and a limit $m$ on the maximum number of elements to add to the vector $\bmg_1$
    \Ensure All vectors of the form $\bmg = \bmg_1 \concat \bmg_2$ where $\ell \leq m$ and $\bmg_2 = (g_{k + 1}, \ldots, g_{k + \ell})$ is a vector of elements in $G$ with the property that $\langle \bmg_1, g_{k + 1}, \ldots, g_{k + j} \rangle < \langle \bmg_1, g_{k + 1}, \ldots, g_{k + j + 1} \rangle$ for all $0 \leq j < \ell$ are added to the set $A$
    \Function{\insexts}{$\bmg_1, G, A, j, m$}
      \If{$\langle \bmg_1 \rangle < G$ and $j < m$}
        \For{$g \in G \setminus \langle \bmg_1 \rangle$}
          \State \Call{\insexts}{$\bmg_1 \concat (g), G, A, j + 1, m$}
        \EndFor
      \Else
        \State Insert $\bmg_1$ into $A$
      \EndIf
    \EndFunction
  \end{algorithmic}
  \caption{An algorithm for computing all extensions of a generating set}
  \label{alg:ins-exts}
\end{algorithm}

Combined with \lemref{gen-coll}, this yields our deterministic square-root speedup for general group isomorphism.

\gengroupiso*

\begin{proof}
  In order to prove this result, we show that \algref{gen-AB} satisfies the conditions of \lemref{gen-coll}.  It is clear that if \algref{gen-AB} returns before reaching \linref{set-A'} then it returns the correct result.  From now on, suppose that it reaches \linref{set-A'}.  In this case, we see that \algref{gen-AB} returns a set $A$ containing ordered generating sets for $G$ and a set $B$ containing ordered generating sets for $H$.  Thus, conditions (a) and (b) are satisfied for some $t(n)$.  We also know that every generating set for the groups $G$ and $H$ has size greater than $(1 / 2) \log_p n$.

  Suppose that $\phi : G \ra H$ is an isomorphism and consider condition (c) of \lemref{gen-coll}.  Consider the value of $\bmh_1 = (h_1, \ldots, h_d)$ after the loop on \linref{h1-loop} terminates.  Since $\bmh_1$ is of length $d$, and satisfies the property that $\langle h_1, \ldots, h_k \rangle < \langle h_1, \ldots, h_{k + 1} \rangle$ for all $k < d$, the $\Call{\insexts}{}$ call on \linref{exts-G} will insert some $\bmg_1 = (g_1, \ldots, g_d)$ into the set $A'$ such that $\phi(\bmg_1) = \bmh_1$.  The loop on \linref{g1-loop} constructs a new set $A$ from $A'$ by extending each vector $\bmg_1$ in $A'$ to an ordered generating set $\bmg = \bmg_1 \concat \bmg_2$ where $\bmg_2 = (g_{d + 1}, \ldots, g_\ell)$ is such that $\langle g_1, \ldots, g_k \rangle < \langle g_1, \ldots, g_{k + 1} \rangle$ for all $k < \ell$.  Then the call to $\Call{\insexts}{}$ on \linref{exts-H} will insert some $\bmh = \bmh_1 \concat \bmh_2$ into the set $B$ such that $\phi(\bmg) = \bmh$.  Thus, condition (c) of \lemref{gen-coll} is satisfied.

  Since all the conditions of \lemref{gen-coll} are satisfied, we obtain an algorithm for group isomorphism.  The only issue that remains is to determine its runtime.  Lines \ref{line:set-p} -- \ref{line:end-gen-set-if} can be performed using brute force in $n^{d + O(1)}$ time.  The call to $\Call{\insexts}{}$ on \linref{exts-G} also takes $n^{d + O(1)}$ time.  The loop on \linref{g1-loop} takes polynomial time for each element of the set $A'$ which contains at most $n^d$ elements; thus it requires at most $n^{d + O(1)}$ time.  The loop on \linref{h1-loop} takes only polynomial time.  Finally, the call to $\Call{\insexts}{}$ on \linref{exts-H} again takes at most $n^{d + O(1)}$ time.  The runtime of \algref{gen-AB} is therefore $n^{d + O(1)} = n^{(1 / 2) \log_p n + O(1)}$.  By \lemref{gen-coll}, the overall algorithm for group isomorphism also takes $n^{(1 / 2) \log_p n + O(1)}$ time.
\end{proof}

\section{Bidirectional collision detection for ring isomorphism}
\label{sec:ring-iso}
We now consider the \emph{ring isomorphism problem}.  Here the rings are specified by tables describing the addition and multiplication relations and we must decide if they are isomorphic.  As usual, we assume that the two rings have the same order and let $n$ denote the number of elements in each of the rings.  While addition in a ring is always commutative, we do not require that the multiplication operation is commutative.

We start by describing a naive method for solving ring isomorphism.  Let $R$ and $S$ be rings and let $p$ be the smallest prime dividing the order of $R$.  We fix an ordered generating set $\bmr$ of size at most $\log_p n$ for the additive group of the ring $R$ and then consider each ordered generating set $\bms$ of size at most $\log_p n$ for the additive group of the ring $S$.  For each such $\bms$, we check if the mapping from $\bmr$ to $\bms$ induces ring isomorphism between $R$ and $S$.  Since we have the naive bound of at most $n^{\log_p n}$ choices for $\bms$, this yields an $n^{\log_p n + O(1)}$ algorithm for testing isomorphism of rings.  By applying bidirectional collision detection, we obtain a square-root speedup which yields an $n^{(1 / 2) \log_p n + O(1)}$ algorithm for testing isomorphism of rings.

\begin{theorem}
  \label{thm:d-ring-iso}
  Ring isomorphism is in $n^{(1 / 2) \log_p n + O(1)}$ deterministic time where $p$ is the smallest prime dividing the order of the ring.
\end{theorem}

Replacing bidirectional collision detection with quantum claw detection~\cite{bassard1997b} yields an $n^{(1 / 3) \log_p n + O(1)}$ time quantum algorithm for testing ring isomorphism.

\begin{theorem}
  \label{thm:q-ring-iso}
  Ring isomorphism is in $n^{(1 / 3) \log_p n + O(1)}$ quantum time where $p$ is the smallest prime dividing the order of the ring.
\end{theorem}

\section{Bidirectional collision detection for $p$-groups and solvable groups}
\label{sec:p-sol-algorithm}
In this section, we show deterministic fourth-root speedups over the generator-enumeration algorithm for $p$-groups and solvable groups; this improves on the author's previous work~\cite{rosenbaum2012b} which gave a square-root speedup.  We start by formulating the conditions that must be satisfied in order for us to obtain a fourth-root speedup.  To this end, we present the following lemma which is a consequence of results in~\cite{rosenbaum2012b}.

\begin{lemma}[\cite{rosenbaum2012b}]
  \label{lem:p-sol-coll}
  Let $G$ and $H$ be $p$-groups and suppose that we have a deterministic algorithm with the following properties:

  \begin{enumerate}
  \item The algorithm either decides if $G \cong H$ or outputs two sets $A$ and $B$ that contain composition series for $G$ and $H$ respectively.
  \item The algorithm runs in time $n^{c \log_p n + O(1)}$ for some constant $c > 0$.
  \item If $\phi : G \ra H$ is an isomorphism, then there exist composition series $S$ in $A$ and $S'$ in $B$ such that $\phi$ is an isomorphism from $S$ to $S'$.
  \end{enumerate}

  Then we can construct deterministic algorithms that can decide isomorphism of

  \begin{itemize}
  \item $p$-groups in $n^{c \log n + O(1)}$ time.
  \item solvable groups in $n^{c \log_p n + O(\log n / \log \log n)}$ time where $p$ is now the smallest prime dividing $n$.
  \end{itemize}
\end{lemma}

The main difference here from the analogous \lemref{gen-coll} for general groups is the requirement in condition (b) that the algorithm runs in $n^{c \log_p n + O(1)}$ time.  While similar results hold when the algorithm satisfies the much weaker condition that its runtime is submultiplicative, the assumption we have chosen allows us to obtain a slightly better runtime for $p$-groups.  We remark that \lemref{p-sol-coll} is much more difficult to prove than \lemref{gen-coll} (see~\cite{rosenbaum2012b}).  For the present paper, it suffices to focus on obtaining an algorithm that satisfies the conditions of \lemref{p-sol-coll}.  We accomplish this by applying bidirectional collision detection to the process of constructing a composition series that we used in~\cite{rosenbaum2012b} to obtain a square-root speedup over the generator-enumeration algorithm.  We start by reviewing this method. 

Consider a group $G$.  The socle of $G$ is a direct product of the minimal normal subgroups of $G$ which are in turn direct products of simple groups.  These simple groups are minimal normal subgroups of the socle so we see that the socle of $G$ can be expressed as a direct product of its simple minimal normal subgroups.  By adding a single simple minimal normal subgroup at each step, we obtain a composition series for the socle of the following form.

\begin{equation}
  \label{eq:mns-comp-series}
  1 \tril S_1 \tril \cdots \tril \bigtimes_{i = 1}^k S_i = \soc(G)
\end{equation}

where each $S_i$ is a simple minimal normal subgroup of the socle of $G$.  An important insight of~\cite{rosenbaum2012b} is that one only has to choose each intermediate product $\bigtimes_{i = 1}^j S_i$ of simple minimal normal subgroups of the socle rather than choosing the simple minimal normal subgroups themselves.  One can then show that the number of composition series of the form of \eqref{mns-comp-series} is at most $\abs{\soc(G)}^{(1 / 2) \log_p \abs{\soc(G)} + O(1)}$ where $p$ is the smallest prime dividing the order of $G$ and $s = \abs{\soc(G)}$.  To extend this bound to composition series for the full group $G$, we first compute a composition series for $\soc(G)$ and recursively compute a composition series for $G / \soc(G)$; we then obtain a composition series for $G$ by lifting the composition series for $G / \soc(G)$ back to subgroups of $G$ using the canonical map.  One then shows that the number of composition series for $G$ that can be constructed in this manner is at most $n^{(1 / 2) \log_p n + O(1)}$.

Let $G$ and $H$ be $p$-groups.  In order to obtain an algorithm that satisfies the conditions of \lemref{p-sol-coll} and is as efficient as possible, we reformulate the problem of constructing a composition series in terms of the bidirectional collision detection notation of \secref{collision}.  As before, we only need to consider the case where $G \cong H$.  We again think of $G$ and $H$ as groups which are isomorphic to an abstract group $K$ via the isomorphisms $\phi_a : K \ra G$ and $\phi_b : K \ra H$.  The tree $T$ now corresponds to the choices made in constructing a composition series for $K$ as described in the previous paragraph.  The root node corresponds to the series $1$; its children are the series $1 \tril S_1$ where $S_1$ is a simple minimal normal subgroup of $\soc(K)$.  Consider the node in $T$ that corresponds to the series $1 \tril S_1 \tril \cdots \tril \bigtimes_{i = 1}^j S_i$ where each $S_i$ is a simple minimal normal subgroup of $\soc(K)$; these are the series $1 \tril S_1 \tril \cdots \tril \bigtimes_{i = 1}^{j + 1} S_i$ where $S_{j + 1}$ is a simple minimal normal subgroup of $\soc(K)$.  We remark that there is only one node for different simple minimal normal subgroups of the socle that result in the same series; in other words, if $\bigtimes_{i = 1}^j S_i \times S_{j + 1} = \bigtimes_{i = 1}^j S_i \times S_{j + 1}'$ then the series $1 \tril S_1 \tril \cdots \tril \bigtimes_{i = 1}^j S_i \times S_{j + 1}$ and $1 \tril S_1 \tril \cdots \tril \bigtimes_{i = 1}^j S_i \times S_{j + 1}'$ are equal and therefore correspond to the same node in $T$.  The children of a node that corresponds to a composition series $1 \tril S_1 \tril \cdots \tril \bigtimes_{i = 1}^k S_i = \soc(K)$ for the socle of $K$ are defined recursively according to the process of constructing a composition series for $K / \soc(K)$.  Each node is denoted by the series obtained by recursively taking the preimages of the subgroups chosen thus far under the corresponding canonical maps.

As usual, we do not have direct access to $T$ but must rely instead on Alice and Bob's labeling functions.  Consider a series $K_0 = 1 \tril K_1 \tril \cdots \tril K_j$ that corresponds to some node in the tree $T$; we define $\lambda_\ell\left(K_0 = 1 \tril K_1 \tril \cdots \tril K_j\right) = \left(1 \tril \phi_\ell[K_1] \tril \cdots \tril \phi_\ell[K_j]\right)$ for $\ell \in \{a, b\}$.  The distance parameter $d$ is no longer equal to a natural quantity.  Instead, we define it to be the smallest natural number such that the subtree $T_1$ corresponding to the nodes within a distance of $d$ of the root contains at least $n^{(1 / 4) \log_p n}$ nodes.  Since there are at most $n$ simple minimal normal subgroups, each node in $T$ has at most $n$ children and $T_1$ therefore contains at most $f(n) = n^{(1 / 4) \log_p n + O(1)}$ nodes.  One can show that each subtree $T_2$ rooted at a node at distance $d$ from the root also contains at most $g(n) = n^{(1 / 4) \log_p n + O(1)}$ nodes.  Thus, we have an $n^{(1 / 4) \log_p n + O(1)}$ algorithm for computing the sets $A$ and $B$ for \lemref{p-sol-coll}.  Let us denote by $\varphi : G \ra G / \soc(G)$ the canonical map; the pseudocode is given in Algorithms~\ref{alg:comps-A} -- \ref{alg:comp-count}.  $\Call{ALL-CHOICES}{}$ is a routine that takes a call to a nondeterministic algorithm and returns the set of all outputs that can be obtained for some choice of the nondeterministic bits.  $\Call{\mns}{}$ is a subroutine that returns the set of all minimal normal subgroups and $\Call{\sims}{}$ is a subroutine that returns the set of all simple minimal normal subgroups.



The following lemma is essential for proving the correctness of our improved $n^{(1 / 4) \log_p n + O(1)}$ deterministic algorithms for $p$-groups and solvable groups.

\begin{lemma}
  \label{lem:comps-AB-cor}
  Let $G$ and $H$ be groups, let $\phi : G \ra H$ an isomorphism and let $t$ a natural number.  Set $A = \Call{\compsA}{G, t}$ and $B = \Call{\compsB}{H, t}$.  Then there exist composition series $S$ in $A$ and $S'$ in $B$ such that $\phi$ is an isomorphism from $S$ to $S'$ \ifft $G \cong H$.
\end{lemma}

\begin{proof}
  The fact that Algorithms~\ref{alg:comps-A} and \ref{alg:comps-B} return sets of composition series already follows from the correctness proof of the algorithm for constructing a composition series from the author's previous paper~\cite{rosenbaum2012b}.  Consider an isomorphism $\phi : G \ra H$ and a natural number $t$ and set $A$ and $B$ accordingly.  Observe that since $A$ is a set of composition series for $G$ and $B$ is a set of composition series for $H$, if $G \not\cong H$ then there are not any isomorphic composition series $S$ and $S'$ in $A$ and $B$.  Thus, we suppose that $G \cong H$ and argue that there exist $S \in A$ and $S' \in B$ such that $\phi$ is an isomorphism from $S$ to $S'$.

  The proof is by induction on the number of recursive calls to $\Call{\comps}{}$.  The base case is when no recursive calls are made; this occurs when either $G$ is simple or $\soc(G) = G$.  If $G$ is simple then so is $H$ since $G \cong H$; we see that $A = \{1 \tril G\}$ and $B = \{1 \tril H\}$ and we are done.  Suppose that $\soc(G) = G$.  We start by considering the call to $\Call{\comps}{}$ in \algref{comps-B}; let $1 \tril S_1' \tril \cdots \tril \bigtimes_{i = 1}^{\bar t} S_i'$ be the prefix of each composition series $S'$ for $H$ which corresponds to the arbitrary choices made on \linref{arb-choose-L} of \algref{comp-count} (note that this is independent of the nondeterministic choices and is therefore well-defined).  Now consider the call to $\Call{\comps}{}$ in \algref{comps-A}.  Since $\phi : G \ra H$ is an isomorphism, one choice for the first simple minimal normal subgroup is $S_1 = \phi^{-1}[S_1']$.  Let us suppose that some combination of nondeterminstic choices on \linref{nondet-choose-L} leads to the prefix $1 \tril S_1 \tril \cdots \tril \bigtimes_{i = 1}^r S_i$ where $r < \bar t$ such that $\phi^{-1}[\bigtimes_{i = 1}^q S_i'] = \bigtimes_{i = 1}^q S_i$ for all $q \leq r$ after the first $r$ choices have been made.  Since $S_{r + 1}'$ is a simple minimal subgroup of $\soc(H)$, $\phi^{-1}[S_{r + 1}']$ is a simple minimal normal subgroup of $\soc(G)$.  Now initially (before any simple minimal normal subgroups were chosen), the set $T$ contained all simple minimal normal subgroups of $\soc(G)$; in particular, it contained $\phi^{-1}[S_{r + 1}']$.  Since lines~\ref{line:set-T'} -- \ref{line:set-T} only remove redundant simple minimal normal subgroups from $T$ and do not affect the set of direct products that can be formed, there exists $S_{r + 1}$ in $T$ such that $\bigtimes_{i = 1}^{r + 1} S_i = \phi^{-1}[\bigtimes_{i = 1}^r S_i'] \times \phi^{-1}[S_{i + 1}'] = \phi^{-1}[\bigtimes_{i = 1}^{r + 1} S_i']$.  It therefore follows by induction that some combination of nondeterministic choices in the call to $\Call{\comps}{}$ in \algref{comps-A} yields a prefix $1 \tril S_1 \tril \cdots \tril \bigtimes_{i = 1}^{\bar t} S_i$ of a composition series $S$ for $G$ such that $\phi^{-1}[\bigtimes_{i = 1}^q S_i'] = \bigtimes_{i = 1}^q S_i$ for all $q \leq \bar t$.  This takes care of all nondeterministic choices in the call to $\Call{\comps}{}$ in \algref{comps-A} so all the remaining choices will be made arbitrarily.  Let $S = (1 \tril S_1 \tril \cdots \tril \bigtimes_{i = 1}^m S_i = \soc(G))$ be the resulting composition series for $G$.

  We now return to the call to $\Call{\comps}{}$ in \algref{comps-B}; recall that we have selected the prefix  $1 \tril S_1' \tril \cdots \tril \bigtimes_{i = 1}^{\bar t} S_i'$ of each composition series $S'$ for $H$.  This prefix corresponds to all arbitrary choices on \linref{arb-choose-L} of \algref{comp-count}.  The remaining choices to be made are all nondeterministic.  As we mentioned before, lines~\ref{line:set-T'} -- \ref{line:set-T} only remove redundant simple minimal normal subgroups and do not affect the set of direct products that can be formed.  It follows that there exists $S_{\bar t + 1}'$ in $T$ such that $\phi[\bigtimes_{i = 1}^{\bar t + 1} S_i] = \bigtimes_{i = 1}^{\bar t + 1} S_i'$; applying this argument inductively shows that for some set of nondeterministic choices yields a composition series $S' = (1 \tril S_1' \tril \cdots \tril \bigtimes_{i = 1}^s S_i' = \soc(H))$ for $H$ such that $\phi$ is an isomorphism from $S$ to $S'$.  This proves the basis case.

  For the inductive case, the argument used in the basis case implies that for some combination of nondeterministic choices, we obtain composition series $1 \tril S_1 \tril \cdots \tril \bigtimes_{i = 1}^m S_i$ and $1 \tril S_1' \tril \cdots \tril \bigtimes_{i = 1}^m S_i'$ for $\soc(G)$ and $\soc(H)$ in both top-level calls to $\Call{\comps}{}$ in Algorithms~\ref{alg:comps-A} and \ref{alg:comps-B} that are isomorphic under $\phi$.  Let us consider the isomorphism $\indphi{G / \soc(G)} : G / \soc(G) \ra H / \soc(H)$ induced by $\phi$ between the factor groups.  The inductive hypothesis implies that for some nondeterministic choices in the first-level recursive calls to $\Call{\comps}{}$, we obtain composition series for $G / \soc(G)$ and $H / \soc(H)$ that are isomorphic under $\indphi{G / \soc(G)}$; the inverse images of these composition series under the canonical maps are also isomorphic under $\phi$.  It follows that we obtain a composition series $S$ for $G$ and a composition series $S'$ for $H$ that are isomorphic under $\phi$ for some combination of nondeterministic choices.
\end{proof}

\begin{algorithm}[H]
  \begin{algorithmic}[1]
    \Input A group $G$ specified by its multiplication table and a natural number $t$
    \Output A set $A$ of composition series for $G$ obtained by choosing the first $t$ simple minimal normal subgroups nondeterministically and choosing the rest arbitrarily
    \Function{$\compsA$}{$G, t$}
      \State \Return $\Call{ALL-CHOICES}{\Call{\comps}{G, 1, t, 0}}$
    \EndFunction
  \end{algorithmic}
  \caption{An algorithm for computing a set $A$ of composition series for $G$}
  \label{alg:comps-A}
\end{algorithm}

\begin{algorithm}[H]
  \begin{algorithmic}[1]
    \Input A group $H$ specified by its multiplication table and a natural number $t$
    \Output A set $B$ of composition series for $H$ obtained by choosing the first $t$ simple minimal normal subgroups arbitrarily and choosing the rest nondeterministically
    \Function{$\compsB$}{$H, t$}
      \State \Return $\Call{ALL-CHOICES}{\Call{\comps}{H, t + 1, \infty, 0}}$
    \EndFunction
  \end{algorithmic}
  \caption{An algorithm for computing a set $B$ of composition series for $H$}
  \label{alg:comps-B}
\end{algorithm}

In order to prove efficiency, it is convenient to make use of the following result from our previous work (which we reformulate to fit the current setting).

\begin{lemma}[\cite{rosenbaum2012b}]
  \label{lem:num-comp-choices}
  Let $G$ be a $p$-group.  Then the number $N$ of composition series returned by $\Call{\compsA}{G, \log_p n}$ or $\Call{\compsB}{H, 0}$ is at most $n^{(1 / 2) \log_p n + O(1)}$.
\end{lemma}

\begin{algorithm}[H]
  \begin{algorithmic}[1]
    \Input A group $G$ specified by its multiplication table, an interval $[a, b]$ and the number $j$ of minimal normal subgroups chosen thus far
    \Output A composition series $S$ for $G$ constructed by choosing simple minimal normal subgroups.  The \nth{i} simple minimal normal subgroup is chosen nondeterministically if $a \leq i \leq b$; otherwise, it is chosen arbitrarily
    \Function{$\comps$}{$G, a, b, j$}
      \State $S \coleq (1, G)$ \inlong{\Comment{The composition series for $G$ will be stored here}}
        \State $\{N_1, \ldots, N_k\} \coleq \Call{\mns}{G}$
        \State $\soc(G) \coleq \genb{N_i}{1 \leq i \leq k}$
        \If{$G$ is not simple}
          \State $T \coleq \bigcup_{i = 1}^k \Call{\sims}{N_i}$ \label{line:comp-T}
          \State $K \coleq 1$
          \While{$K \not= \soc(G)$} \label{line:prod-loop}
            \If{$a \leq j \leq b$}
              \State Nondeterministically choose $L \in T$ \label{line:nondet-choose-L}
            \Else
              \State Arbitrarily choose $L \in T$ \label{line:arb-choose-L}
            \EndIf
            \State $K \coleq K \times L$
            \State Insert $K$ into $S$
            \State $j \coleq j + 1$
            \State $T' = \emptyset$ \label{line:set-T'}
            \For{$L \in T$} \label{line:loop-reps}
              \If{$K \cap L = 1 \text{ and } K \times L \not= K \times L'$ for all $L' \in T'$}
                \State Insert $L$ into $T'$
              \EndIf
            \EndFor \label{line:loop-reps-end}
            \State $T \coleq T'$ \label{line:set-T}
          \EndWhile
          \State Label $\soc(G)$ as ``socle'' in $S$
          \If{$\soc(G) \not= G$}
            \State $(K_0 \tril \cdots \tril K_m) \coleq \Call{\comps}{G / \soc(G), a, b, j}$ \label{line:recur}
            \For{$i = 1, \ldots, m - 1$}
              \State Insert $G_i \coleq \varphi^{-1}[K_i]$ into $S$ and copy the label of $K_i$ to $G_i$
            \EndFor
          \EndIf
        \EndIf
        \State \Return $S$
    \EndFunction
  \end{algorithmic}
  \caption{A subroutine for computing a composition series}
  \label{alg:comp-count}
\end{algorithm}

The following will be used to derive the runtime of our algorithm.

\begin{lemma}
  \label{lem:comp-t}
  Let $G$ be a $p$-group.  Then we can compute a natural number $t(G)$ in polynomial time which has the following properties:

  \begin{enumerate}
  \item $\Call{\compsA}{G, t(G)}$ and $\Call{\compsB}{G, t(G)}$ take $n^{(1 / 4) \log_p n + O(1)}$ time.
  \item If $H$ is a group and $G \cong H$, then $t(G) = t(H)$.
  \end{enumerate}
\end{lemma}

\begin{proof}
  We select $t(G)$ so that $\Call{\compsA}{G, t(G)}$ makes close to $n^{(1 / 4) \log_p n}$ nondeterministic choices.  Let $\ell$ be the number of subgroups which are labelled ``socle'' in the composition series obtained when we run $\Call{\comps-A}{G, 0}$; we define $F_0 = G$ and $F_{i + 1} = F_i / \soc(F_i)$ and observe that the last subgroup labelled ``socle'' is $\soc(F_{\ell - 1})$.  Let $m_i$ be the number of simple minimal normal subgroups required to express $\soc(F_{i - 1})$ as a direct product.  (To see that this is well-defined, separate the socle into its Abelian and non-Abelian simple minimal normal subgroups and recall that there is at most one way to express a group as a direct product of non-Abelian simple groups.)  Let us define $N_i = \prod_{j = 0}^{m_i - 1} (s_i / p^j)$ where $s_i = \abs{\soc(F_{i - 1})}$; this is an upper bound on the number of possible ways to construct the composition series for the \nth{i} socle.  We let $N = \prod_{i = 1}^\ell N_i$; this is an upper bound on the total number of choices in constructing a composition series for the full group $G$.  We choose the largest $r$ such that $\prod_{i = 1}^{r - 1} N_i \leq \sqrt{N}$; we then select the largest $u$ such that $\left(\prod_{i = 1}^{r - 1} N_i\right) \left(\prod_{j = 0}^{u - 1} (s_r / p^j)\right) \leq \sqrt{N}$.  Since $\sqrt{N} < \left(\prod_{i = 1}^{r - 1} N_i\right) \left(\prod_{j = 0}^{u} (s_r / p^j)\right) \leq \left(\prod_{i = 1}^{r - 1} N_i\right) \left(\prod_{j = 0}^{u - 1} (s_r / p^j)\right) n$, we see that $\left(\prod_{i = 1}^{r - 1} N_i\right) \left(\prod_{j = 0}^{u - 1} (s_r / p^j)\right)$ is within a factor of $n$ of $\sqrt{N}$.  Letting $t(G) = \sum_{i = 1}^{r - 1} m_i + u$, the call $\Call{\compsA}{G, t(G)}$ therefore results in $\sqrt{N}$ nondeterministic choices up to a factor of $n$.  Since $N = n^{(1 / 2) \log_p n + O(1)}$ (see~\cite{rosenbaum2012b} and \lemref{num-comp-choices}), this call requires $\sqrt{N} \poly(n) = n^{(1 / 4) \log_p n + O(1)}$ time.  We note that computing $t(G)$ required only polynomial time.

  For the call $\Call{\compsB}{G, t(G)}$, we note that $\left(\prod_{i = 1}^{r - 1} N_i\right) \left(\prod_{j = 0}^{u - 1} (s_r / p^j)\right)$ is within a factor of $n$ of $\sqrt{N}$, so $\left(\prod_{j = u}^{m_r - 1} (s_r / p^j)\right) \left(\prod_{i = r + 1}^\ell N_i\right)$ is also within a factor of $n$ of $\sqrt{N}$.  Since $\left(\prod_{j = u}^{m_r - 1} (s_r / p^j)\right) \left(\prod_{i = r + 1}^\ell N_i\right)$ is an upper bound on the number of nondeterministic choices made in the call $\Call{\compsB}{G, t(G)}$, we see that this call also takes at most $\sqrt{N} \poly(n) = n^{(1 / 4) \log_p n + O(1)}$ time.

  Finally, let $H$ be a group isomorphic to $G$.  We note that $t(G)$ depends only on the smallest prime $p$ that divides the order of $G$, the number $\ell$ of subgroups labelled ``socle'' and the composition length $m_i$ of each $F_i$.  Since all of these are isomorphism invariant, we see that $t(G) = t(H)$.
\end{proof}

The fourth-root speedups now follow easily.

\begin{theorem}
  \label{thm:p-group-iso}
  $p$-group isomorphism is in $n^{(1 / 4) \log n + O(1)}$ deterministic time.
\end{theorem}

\begin{proof}
  We start by computing $t(G)$ and $t(H)$ using \lemref{comp-t}.  If $t(G) \not= t(H)$, then $G \not\cong H$; otherwise, we set $t = t(G) = t(H)$ and compute $A = \Call{\compsA}{G, t}$ and $B = \Call{\compsB}{H, t}$.  By \lemref{comps-AB-cor}, $\phi : G \ra H$ is an isomorphism \ifft there exist composition series $S$ in $A$ and $S'$ in $B$ such that $\phi$ is an isomorphism from $S$ to $S'$.  This takes care of conditions (a) and (c) of \lemref{p-sol-coll}.  For condition (b), we note that by \lemref{comp-t}, $n^{(1 / 4) \log_p n + O(1)}$ is an upper bound on the runtime of the procedure described so far.  The result is then immediate from \lemref{p-sol-coll}.
\end{proof}

The same argument yields a fourth-root speedup for solvable groups.

\solgroupiso*


\section{A $T^{1 / \sqrt{2}}$ speedup for graph isomorphism}
\label{sec:graph-iso}
We start by recalling the high-level structure of the best known algorithm for graph isomorphism~\cite{babai1983b}.  The first step involves reducing the original graph isomorphism problem to $n^{4 n / d}$ instances of graph isomorphism problems where the graphs have color-degree\footnote{The color-degree of a node is a technical concept that is required for the statements of the theorems reviewed here to be correct.  However, for the purposes of this section, nothing is lost if one simply thinks of the color-degree of a node as the ordinary notion of degree.} at most $d$.

\begin{lemma}[Zemlyachenko, cf.~\cite{babai1981a}]
  \label{lem:zemlyachenko}
  Let $X$ and $Y$ be graphs.  There is a polynomial-time deterministic procedure $P$ that takes a graph and a sequence of vertexes as its input and outputs a colored graph with the following properties:

  \begin{enumerate}
  \item If $X \cong Y$, then for any sequence of nodes $x_1, \ldots, x_m$ in $X$, there exists a sequence of nodes $y_1, \ldots, y_m$ in $Y$ such that $P(X; x_1, \ldots, x_m) \cong P(Y; y_1, \ldots, y_m)$.
  \item If $X \not\cong Y$, then for all sequences of nodes $x_1, \ldots, x_m$ and $y_1, \ldots, y_m$ in $X$ and $Y$, $P(X; x_1, \ldots, x_m) \not\cong P(Y; y_1, \ldots, y_m)$.
  \end{enumerate}

  Moreover, we can deterministically compute in polynomial time a sequence of $4 n / d$ nodes $x_1, \ldots, x_m$ in $X$ such that $P(X; x_1, \ldots, x_m)$ has color-degree at most $d$.
\end{lemma}

To obtain an algorithm for graph isomorphism, we compute a sequence of $4 n / d$ nodes $x_1, \ldots, x_m$ in $X$ such that $P(X; x_1, \ldots, x_m)$ has color-degree at most $d$; we then consider all $n^{4 n / d}$ possible sequences $y_1, \ldots, y_m$  of $4 n / d$ nodes in $Y$ and check if $P(X; x_1, \ldots, x_m) \cong P(Y; y_1, \ldots, y_m)$ for one of these sequences.  This occurs \ifft $X \cong Y$.  By combining with an $n^{c d / \log d}$ algorithm~\cite{babai1983b} for testing isomorphism of graphs of color-degree at most $d$, we obtain an $n^{4 n / d + c d / \log d + O(1)}$ algorithm for graph isomorphism where $d$ is a parameter that we choose.  Minimizing the runtime over $d$ yields the best known algorithm for graph isomorphism.

\begin{theorem}[Babai, Kantor and Luks~\cite{babai1983b}]
  \label{thm:graph-iso}
  Graph isomorphism can be decided in $\exp(O(\sqrt{n \log n}))$ deterministic time.
\end{theorem}

Optimizing the constant in the exponent of the theorem, we obtain a runtime of $2^{(4 \sqrt{2 c}) \sqrt{n \log n} + O(\log n)}$.  Our contribution in this section is to note that bidirectional collision detection can be applied to choose $n^{2 n / d}$ sequences of $4 n / d$ nodes $x_1, \ldots, x_m$ and $y_1, \ldots, y_m$ in both $X$ and $Y$ such that for some pair of sequences, $P(X; x_1, \ldots, x_m) \cong P(Y; y_1, \ldots, y_m)$ \ifft $X \cong Y$.  Some care is required as not all sequences of $4 n / d$ nodes result in a graph of color-degree at most $d$.  However, the polynomial-time algorithm of \lemref{zemlyachenko} (for constructing a sequence of nodes that result in a graph of color-degree at most $d$) can be modified to make certain nondeterministic choices such that every combination of nondeterministic choices results in a graph of color-degree at most $d$.  Moreover, if two graphs $X$ and $Y$ are isomorphic, then every colored graph resulting from some combination of nondeterministic choices for $X$ is isomorphic to a colored graph resulting from some set of nondeterministic choices for $Y$ (and the reverse holds as well).  We use bidirectional collision detection to choose sets of $n^{2 n / d}$ sequences $x_1, \ldots, x_m$ and $y_1, \ldots, y_m$ of $4 n / d$ nodes in both $X$ and $Y$; then we compute the canonical form of each of the graphs $P(X; x_1, \ldots, x_m)$ and $P(Y; y_1, \ldots, y_m)$ in $n^{c d / \log d}$ time~\cite{babai1983a,babai1983b} and sort the results.  This enables us to check if a pair of sequences that result in isomorphic graphs exist in $n^{2 n / d + c d / \log d + O(1)}$ time.  Since it is easy to see that the exponent is minimized up to constant factors when $d = c' \sqrt{n \log n}$, we obtain a runtime of $2^{2 (1 / c' + c c') \sqrt{n \log n} + O(\log n)}$.  Since $c' > 0$ is a constant that we choose, we can further optimize the exponent by minimizing $2 (1 / c' + c c')$ with respect to $c'$.  This yields a final runtime of $2^{4 \sqrt{c} \sqrt{n \log n} + O(\log n)}$.  We summarize this result in the following theorem.

\begin{theorem}
  \label{thm:d-bd-graph-iso}
  Graph isomorphism can be decided in $2^{4 \sqrt{c} \sqrt{n \log n} + O(\log n)}$ deterministic time where $c$ is the constant in the exponent of the $n^{c d / \log d}$ time algorithm~\cite{babai1983a,babai1983b} for computing canonical forms of graphs of color-degree at most $d$.
\end{theorem}

This is a $T^{1 / \sqrt{2}}$ speedup over the previous best runtime of $2^{(4 \sqrt{2 c}) \sqrt{n \log n} + O(\log n)}$.



\section{Time-space tradeoffs}
\label{sec:time-space}
In \secref{collision}, we mentioned time-space tradeoffs for collision detection problems that fit into the tree framework introduced in that section.  We now apply this idea to obtain time-space tradeoffs for group isomorphism problems.  We start by noting that although the bidirectional generator-enumeration algorithm of \secref{gen-algorithm} solves general group isomorphism in only $n^{(1 / 2) \log_p n + O(1)}$ time where $p$ is the smallest prime dividing the order of the group, it requires $n^{(1 / 2) \log_p n + O(1)}$ space; on the other hand, the generator-enumeration algorithm can be implemented in polynomial space but requires $n^{\log_p n + O(1)}$ time.  Both of these algorithms may be regarded as extreme points of a more general deterministic time-space tradeoff obtained using the argument of \ssecref{d-time-space}.  If one is willing to allow quantum computation, then more efficient time-space tradeoffs are possible.  The same arguments can also be applied in order to obtain time-space tradeoffs for solvable-group isomorphism and the other isomorphism problems considered in this paper.

Let $G$ and $H$ be groups such that $p$ is the smallest prime dividing the order of $G$.  Using \algref{gen-AB}, we can compute sets $A$ and $B$ of size at most $n^{(1 / 2) \log_p n + O(1)}$ in time $n^{(1 / 2) \log_p n + O(1)}$ that contain a common element \ifft $G \cong H$.  The time space tradeoff for general groups then follows immediately by the argument of \ssecref{d-time-space}.

\begin{theorem}
  Let $G$ and $H$ be groups where $p$ is the smallest prime dividing the order of the group.  Then isomorphism can be tested deterministically in $O(T)$ time and $O(S)$ space where

  \begin{equation}
    TS = n^{\log_p n + O(1)}
  \end{equation}

  and $S \leq n^{(1 / 2) \log_p n + O(1)}$.
\end{theorem}

The quantum algorithm follows via the argument of \ssecref{q-time-space}.

\begin{theorem}
  Let $G$ and $H$ be groups where $p$ is the smallest prime dividing the order of the group.  Then isomorphism can be tested using quantum algorithms in $O(T)$ time and $O(S)$ space where

  \begin{equation}
    T \sqrt{S} = n^{(1 / 2) \log_p n + O(1)}
  \end{equation}

  and $S \leq n^{(1 / 4) \log_p n + O(1)}$.
\end{theorem}

For solvable-group isomorphism, we obtain deterministic and quantum time-space tradeoffs using the same techniques.

\begin{theorem}
  Let $G$ and $H$ be solvable groups where $p$ is the smallest prime dividing the order of the group.  Then isomorphism can be tested deterministically in $O(T)$ time and $O(S)$ space where

  \begin{equation}
    TS = n^{(1 / 2) \log_p n + O(\log n / \log \log n)}
  \end{equation}

  and $S \leq n^{(1 / 4) \log_p n + O(\log n / \log \log n)}$.
\end{theorem}

\begin{theorem}
  Let $G$ and $H$ be solvable groups where $p$ is the smallest prime dividing the order of the group.  Then isomorphism can be tested using quantum algorithms in $O(T)$ time and $O(S)$ space where $S \leq n^{(1 / 6) \log_p n + O(\log n / \log \log n)}$ and

  \begin{equation}
    T \sqrt{S} = n^{(1 / 4) \log_p n + O(\log n / \log \log n)}
  \end{equation}
\end{theorem}

Similar results hold for rings and graphs.

\section{Conclusion}
\label{sec:conclusion}
In this work, we introduced the bidirectional collision-detection technique and used it to obtain speedups over the previous best algorithms for the group, ring and graph isomorphism problems.  We summarize the state of the art for the isomorphism problems considered in this paper in~\tabref{group-iso}.

\begin{table}[H]
  \centering
  \begin{tabulary}{\textwidth}{|C|C|C|C|}
    \hline
    \textbf{Class of objects} & \textbf{Runtime} & \textbf{Paradigm} & \textbf{Speedup} \\
    \hline
    General groups & $n^{(1 / 2) \log n + O(1)}$ & Deterministic & $T^{1 / 2}$ \\ 
    \hline
    General groups & $n^{(1 / 3) \log n + O(1)}$ & Quantum & $T^{2 / 3}$ \\ 
    \hline
    Solvable groups & $n^{(1 / 4) \log n + O(\log n / \log \log n)}$ & Deterministic & $T^{1 / 4}$ \\ 
    \hline
    Solvable groups & $n^{(1 / 6) \log n + O(\log n / \log \log n)}$ & Quantum & $T^{1 / 3}$ \\ 
    \hline
    $p$-groups & $n^{(1 / 4) \log n + O(1)}$ & Deterministic & $T^{1 / 4}$ \\ 
    \hline
    $p$-groups & $n^{(1 / 6) \log n + O(1)}$ & Quantum & $T^{1 / 3}$ \\ 
    \hline
    Rings & $n^{(1 / 2) \log n + O(1)}$ & Deterministic & $T^{1 / 2}$ \\ 
    \hline
    Rings & $n^{(1 / 3) \log n + O(1)}$ & Quantum & $T^{2 / 3}$ \\ 
    \hline
    Graphs & $2^{4 \sqrt{c} \sqrt{n \log n} + O(\log n)}$ & Deterministic & $T^{1 / \sqrt{2}}$ \\ 
    \hline
  \end{tabulary}
  \caption{Algorithms for isomorphism problems}
  \label{tab:group-iso}
\end{table}

It is interesting to note that there is currently no advantage for randomized algorithms over deterministic algorithms in this regime.  We consider the question of whether such algorithms exist to be an interesting open problem; the techniques used in the author's previous work~\cite{rosenbaum2012b} for constructing faster randomized algorithms no longer suffice so new ideas appear to be required.

Our algorithm for general groups relies only on bidirectional collision detection and does not exploit the composition series machinery used in the algorithms for $p$-groups and solvable groups.  It would certainly be interesting if a method could be found that allowed us to obtain a square-root speedup using composition series for general groups as well; by combining with bidirectional collision detection, we could then hope for a fourth-root speedup for general groups.

\section*{Acknowledgements}
I thank Paul Beame and Aram Harrow for useful discussions and feedback.  Part of this work was completed while I was a long term visitor at the Center for Theoretical Physics at the Massachusetts Institute of Technology.  I was funded by the DoD AFOSR through an NDSEG fellowship.  Partial support was provided by the NSF under grant CCF-0916400.

\bibliographystyle{initials}
\bibliography{$HOME/LaTeX/computer-science-references,$HOME/LaTeX/math-references,$HOME/LaTeX/quantum-computing-references} 

\end{document}